\documentclass[conference]{IEEEtran}

\IEEEoverridecommandlockouts
\usepackage{graphicx,cite,calc,xcolor,amssymb,amsmath,mathrsfs,dsfont,epstopdf}
\usepackage[normalem]{ulem}

\usepackage{float}

\usepackage{graphicx}
\usepackage{subcaption}
\usepackage{amsmath}
\usepackage{lipsum}
\usepackage{algorithm}
\usepackage{algpseudocode}
\usepackage{soul} 
\usepackage{xcolor}
\usepackage{optidef}
\usepackage{listings}
\usepackage{enumitem}
\usepackage{amsthm}
\usepackage[T1]{fontenc}
\usepackage{caption}

\usepackage{algorithm}
\makeatletter
\newcommand*{\rom}[1]{\expandafter\@slowromancap\romannumeral #1@}
\makeatother
\usepackage {amsmath}
\makeatletter
\newtheorem{proposition}{Proposition}
\usepackage{newtxtext,newtxmath}

\usepackage{booktabs}
\usepackage{adjustbox} 

\begin{document}


\title{A Novel 3D Antenna Architecture with Spatial Resource Allocation for Massive MIMO HAPS }

\author{\IEEEauthorblockN{Rozita Shafie\IEEEauthorrefmark{1}\IEEEauthorrefmark{2},
Omid Abbasi\IEEEauthorrefmark{1}, 
Halim Yanikomeroglu\IEEEauthorrefmark{1}, and
Mohammad Javad Omidi\IEEEauthorrefmark{2}\IEEEauthorrefmark{3}}
 \thanks{This work was sponsored, in part, by a Natural Sciences and Engineering Research Council of Canada (NSERC) Discovery Grant. }
\IEEEauthorblockA{\IEEEauthorrefmark{1}Non-Terrestrial Networks Lab, Department of Systems and Computer Engineering, Carleton University, Ottawa, Canada\\
\IEEEauthorrefmark{2}Department of Electrical and Computer Engineering, Isfahan University of Technology, Isfahan 84156-83111, Iran\\
\IEEEauthorrefmark{3}Department of Electronics and Communication Engineering, Kuwait College of Science and Technology, Doha 35003, Kuwait\\
Email: \IEEEauthorrefmark{1}rozitashafie@sce.carleton.ca,
\IEEEauthorrefmark{1}omidabbasi@sce.carleton.ca,
\IEEEauthorrefmark{1}halim@sce.carleton.ca,
\IEEEauthorrefmark{2}omidi@iut.ac.ir}}

\maketitle



\begin{abstract}

Spatial correlation poses a significant challenge in massive multiple-input multiple-output (MIMO) high-altitude platform station (HAPS) systems. The inherent spatial correlation among antenna elements on the HAPS induces high correlation and interference among users' channel gains. To mitigate this issue, we propose an integrated approach that combines spatial resource allocation and user clustering.
 In our proposed solution, we assign the same resource blocks to users with orthogonal channel gains, while users with non-orthogonal channel gains receive different resource blocks.
 Additionally, we propose a sectorized antenna architecture for the massive MIMO HAPS base station, specifically designed to directly transmit three-dimensional beams to users and reduce spatial correlation among antenna elements.
This work addresses the joint optimization problem of power allocation and resource allocation to maximize the overall data rate of the massive MIMO HAPS system. Simulation results revealed the role of spatial resource allocation in managing spatial correlation and interference among users.
\end{abstract}


\section{Introduction}
The progression towards 6G wireless networks necessitates innovative technologies to support higher data rates, reduced latency, seamless connectivity, and broader coverage. High-altitude platform station (HAPS) systems present a promising approach to achieving these goals \cite{9380673}. 
By leveraging the vast antenna arrays available on HAPS, massive multiple-input multiple-output (MIMO)  has the potential to significantly improve spectral efficiency, capacity, and interference management \cite{9779715,10736960, 11008878, 10608095}.  However, the integration of massive MIMO to HAPS introduces a critical challenge – the issue of spatial correlation \cite{10008738}.  Unlike traditional terrestrial deployments where spatial correlation is often minimal, HAPS, due to their altitude and geometry, lead to correlated channels among antennas, potentially compromising the performance gains anticipated from Massive MIMO. Addressing the spatial correlation issue in the massive MIMO HAPS systems requires innovative approaches. These include optimizing spatial resource allocation, developing advanced beamforming techniques, and designing innovative antenna architectures \cite{9110855}. 

 Designing antenna architectures can significantly reduce the adverse effects of spatial correlation among antenna elements, thus enhancing the performance of MIMO systems on HAPS. 
 In addition to antenna design, spatial resource allocation stands out as a crucial approach for addressing spatial correlation within massive MIMO systems deployed on HAPS. While antenna design shapes the physical characteristics of the array, spatial resource allocation complements this by strategically distributing communication resources such as transmit power, beamforming weights, time and frequency resources, and user assignments across the array of antennas. Spatial resource allocation exploits the spatial diversity inherent to the antenna array, effectively mitigating the adverse effects of correlated channels by intelligently allocating these resources~\cite{9110855}.

 This research introduces a novel sectorized antenna architecture that integrates uniform planar arrays  (UPAs) to mitigate spatial correlation in antenna elements and enable 3-dimensional (3D) beam generation.  
Additionally, we have introduced a spatial resource allocation technique for terrestrial users aimed at alleviating spatial correlation and interference issues among them. Our approach involves assigning identical resource blocks to users whose channel gains are orthogonal to each other, while allocating distinct resource blocks to users with non-orthogonal channels. 
In the process of user grouping, we utilize the directional information of users' channels to identify those with orthogonal channels to each other. 
Then we determine the optimal number of resource blocks required to serve the users effectively.
This paper addresses the joint optimization problem of power and resource allocation. Within this context, we strive to efficiently allocate power and resource blocks among users to maximize the overall data rate while ensuring that users' quality of service (QoS) requirements are met.
\section{System Model}
\begin{figure*}[htbp]
\begin{centering}
\centerline{\includegraphics[width = 16.35cm ]{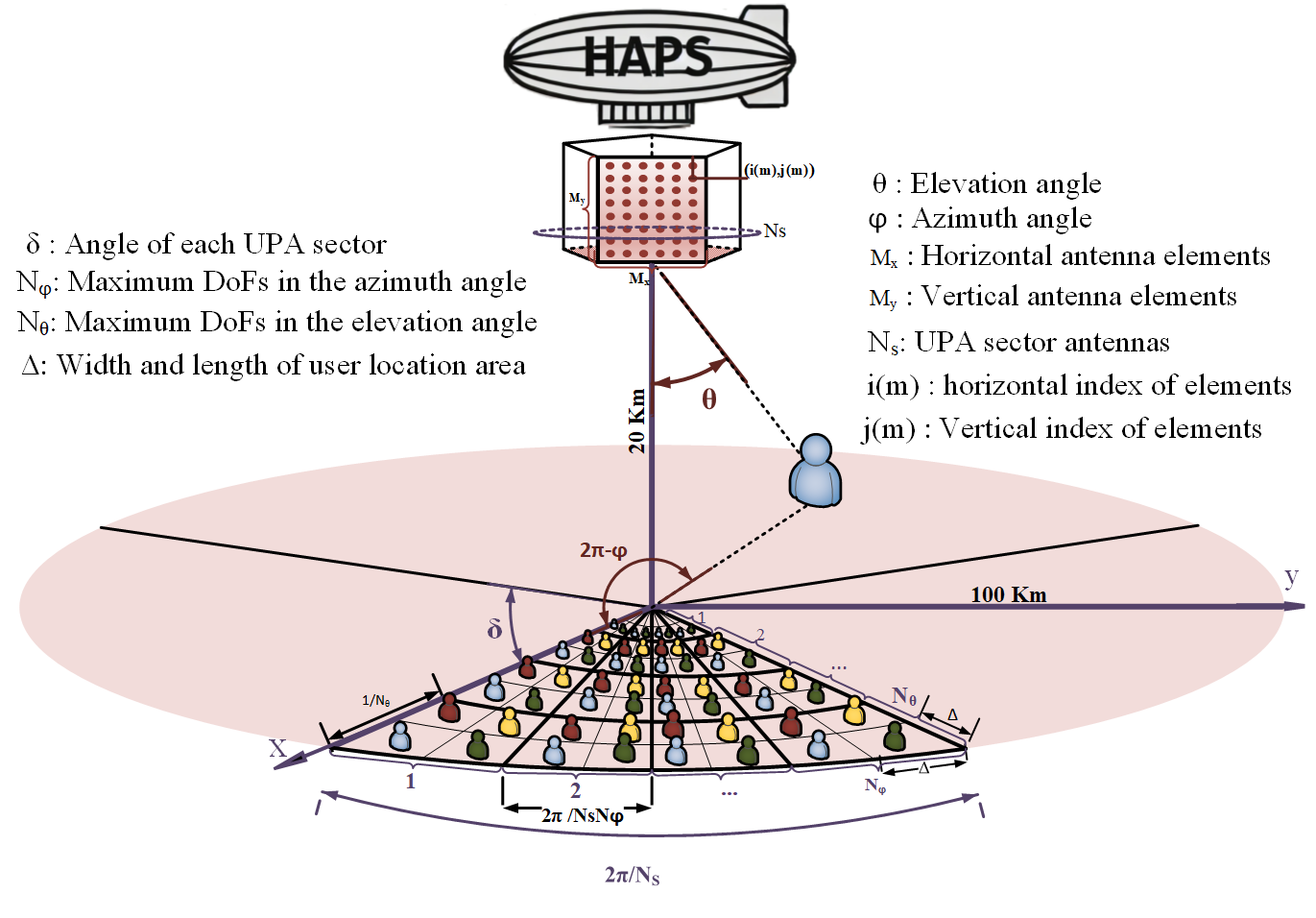}}
\caption{The proposed antenna architecture featuring a sectorized cylindrical design and corresponding communication setups. In this figure, users with the same color belong to the same group and are allocated the same resources. The scenario features \( N_\theta \) sections in radius and \( N_\varphi \) sections in azimuth for each UPA sector, with a frequency reuse factor of $4$ for each section.}
\label{system_model_figure}
\end{centering}
\end{figure*}
In this study, we investigate a network that encompasses terrestrial users and a massive MIMO HAPS system. This system model is depicted in Fig.~\ref{system_model_figure}. The HAPS system hovers at a strategic altitude of approximately $20$ kilometers, enabling a wide coverage area that extends to a radius of $100$ kilometers for ground-based users. To maintain simplicity, each user is equipped with a single omnidirectional antenna. To efficiently manage the network resources, terrestrial users are organized into distinct clusters, each assigned specific resource blocks. Within a given cluster, all users operate on the same frequency band, while different clusters utilize different frequency bands. 
The azimuth ($\varphi$) and elevation ($\theta $) angles are both defined as the angle with the x-axis in the horizontal plane and the angle in front of the HAPS elevation, respectively. 

\section{Proposed antenna architecture for HAPS}
 Fig.~\ref{system_model_figure} illustrates the proposed antenna architecture for the HAPS, which incorporates a configuration comprising vertical UPAs positioned around the cylindrical. Each sector of the UPA comprises \( M_x \times M_y \) antenna elements. The elements \(M_x\) and \(M_y\) are positioned in the horizontal and vertical planes, respectively. The spacing between the antenna elements is a critical parameter. In our scheme, \( d_H \) dictates the horizontal separation and \( d_V \) the vertical separation of the elements that are carefully chosen to be integral multiples of the signal wavelength.
The antenna index of the UPA is denoted by $m \in \left[ {1,M} \right] $. In general, for the location of the $m^{\rm th}$ antenna in the $n^{\rm th}$ sector, we have \cite{SIG-093}
\begin{equation}
{\bf u}_m^n   = \left[ {\begin{array}{*{20}c}
    0\\
   {i\left( m \right)d_H \lambda } \\
   {j\left( m \right)d_V \lambda }   \\
\end{array}} \right],m \in \left[ {1,M} \right] \\  ,\; \; \; \; n=1,...,N_s,
\end{equation}
\begin{equation}
\begin{array}{l}
 i\left( m \right) = \bmod \left( {m - 1,M_x } \right), \\ 
 j\left( m \right) = \left\lfloor {{{\left( {m - 1} \right)} \mathord{\left/
 {\vphantom {{\left( {m - 1} \right)} {M_x }}} \right.
 \kern-\nulldelimiterspace} {M_x }}} \right\rfloor . \\ 
 \end{array}
\end{equation}

\subsection{Propagation and Channel Modeling}
The HAPS base station is situated at a high elevation from the ground. As a result, scatterers are rarely found near the base station. However, the scatterers are generally located around the users and the signal distribution in the environment is non-uniform. Due to the non-uniform distribution of the signal in the HAPS propagation environment, and the existence of an LoS path between the HAPS and terrestrial users, the channel for user $k$ is modeled with a Gaussian distribution  as shown below:
\begin{equation}
{\bf h}_k   \sim {\rm N} \left( {{\bf \bar h}_k  ,{\bf C}_k   }
\right),
\end{equation}
where ${\bf \bar h}_k  $ corresponds to the line of sight (LoS) component and  ${\bf C}_k   $  is a positive semi-definite covariance matrix that describes the spatial correlation of the non (N)-LoS components. In other words, it describes the statistical relationship between the signal fading at different antenna positions\cite{8620255}.
With $M$ antenna elements on  UPA, the LoS channel response ${\bf \bar h}_k\in \mathbb{C}^{M} $ for    user $k$ is given by
\begin{equation}
{\bf \bar h}_k  = \sqrt{\beta _k^{\rm{LoS}} }\cdot ({\bf v}_k^\varphi \otimes {\bf v}_k^h),
\end{equation}
where the symbol \( \otimes \) denotes the Kronecker product.   The vectors $ {\bf v}_k^\varphi$ and ${\bf v}_k^h $ represent the azimuth and elevation direction vectors, respectively, that can be derived as
\begin{equation}
{\bf v}_k^\varphi   = \frac{1}{{\sqrt {M_x  } }}\left[ {\begin{array}{*{20}c}
   {1} & {e^{ - j\pi \mu _k^\varphi  } } & { \cdots } & {e^{ - j\pi \left( {M_x   - 1} \right)\mu _k^\varphi  } }  \\
\end{array}} \right]^T  \in C^{M_x   \times 1} ,
\label{aa1}
\end{equation}
\begin{equation}
{\bf v}_k^h  = \frac{1}{{\sqrt {M_y } }}\left[ {\begin{array}{*{20}c}
   1 & {e^{ - j\pi \mu _k^h } } &  \cdots  & {e^{ - j\pi \left( {M_y  - 1} \right)\mu _k^h } }  \\
\end{array}} \right]^T  \in C^{M_y  \times 1} .
\label{bb2}
\end{equation}
In these equations, $\mu_k^\varphi$ and $\mu_k^h$ are spatial angular coordinates for user $k$ that are related to the physical angles as $\mu_{k}^\varphi = \sin(\theta_{k}) \cos(\varphi_{k})$ and $\mu_{k}^h = \cos(\theta_{k})$.
Based on the Kronecker multiplication and simplification of the LoS channel equation, we are able to derive the following formula:
\begin{equation}
{\bf \bar h}_k   = \sqrt{\beta _k^{\rm{LoS}} } \left[ {e^{j{\bf k}\left( {\varphi _k ,\theta _k } \right)^T {{\bf u}_1  }}} , \cdots ,e^{j{\bf k}\left( {\varphi _k ,\theta _k } \right)^T {{\bf u}_M }} \right]^T,
 \label{eqa2}
 \end{equation}
 where  ${\varphi _k} $ and $\theta _k$ are  the azimuth angle and the elevation angle of user $k$ to the HAPS antenna, respectively.
 The wave vector  $ {\bf k}\left( {\varphi _k ,\theta _k } \right)$
 is also defined as follows\cite{SIG-093}:
\begin{equation}
{\bf k}\left( {\varphi _k ,\theta _k } \right) = \frac{{2\pi }}{\lambda }\left( {\begin{array}{*{20}c}
   {\cos \left( {\theta _k } \right)\cos \left( {\varphi _k } \right)}  \\
   {\cos \left( {\theta _k } \right)\sin \left( {\varphi _k } \right)}  \\
   {\sin \left( {\theta _k } \right)}  \\
\end{array}} \right) .
 \end{equation}
The terms \( \beta_k^{\text{LoS}} \) and \( \beta_k^{\text{NLoS}} \) represent the large-scale fading coefficients for the LoS and  NLoS paths, respectively \cite{3gpp.36.33}. 

Referring to the 3D local scattering model and the assumption of a uniform distribution for the independent random variables $\theta$ and $\varphi$, we can express the $(a,b)^{th}$ entity of the spatial correlation matrix (${\bf C}_k \in \mathbb{C}^{M \times M}$) as follows \cite{10008738, SIG-093}:
 \begin{equation}
 \left[ {{\bf C}_k  } \right]_{a,b}  = \frac{{\beta _k^{\rm NLoS , n} }}{{4\Delta \varphi \Delta \theta }}\int_{\theta  - \Delta \theta }^{\theta  + \Delta \theta } {\int_{\varphi  - \Delta \varphi }^{\varphi  + \Delta \varphi } {e^{j{\bf k}\left( {\varphi _k ,\theta _k } \right)^T \left( {{\bf u}_a  - {\bf u}_b } \right)} d\varphi d\theta,} }
 \end{equation}
The horizontal angular
spread $\left( {\Delta \varphi } \right)$ and the vertical spread $\left( {\Delta \theta } \right)$ 
are related to the 3D one-ring model \cite{SIG-093}.

\section{Proposed user grouping and resource allocation}
In our HAPS network communication system design, leveraging user channel gain orthogonality is essential for enhanced performance. We propose using user channel direction vectors as a precoding matrix to align transmitted signals with channel directionality, enabling effective signal isolation upon reception.
The channel direction vectors, denoted as \({\bf v}_k ={\bf v}_k^\varphi \otimes {\bf v}_k^h \) for the \(k\)-th user, are derived from the combination of azimuth and elevation information, as detailed in equations \eqref{aa1} and \eqref{bb2}. 

In the proposed integrated communication system, terrestrial users are organized into clusters to optimize resource allocation and enhance system efficiency. This clustering approach is underpinned by the spatial characteristics of users' channel directions. Specifically, users that exhibit orthogonal channel directions are aggregated into the same cluster. This orthogonality ensures minimal interference among the signals of users within the same group, allowing them to share the same set of resource blocks effectively.
The orthogonality condition between the channel vectors of users \(i\) and \(j\) is mathematically represented as \({\bf v}_i^H {\bf v}_k = 0\) for \(i \neq k\), where \({\bf v}_i\) and \({\bf v}_j\) are the channel direction vectors of users \(i\) and \(j\), respectively, ensuring minimal interference between them.
By expanding the orthogonality condition and applying the properties of the Kronecker product along with the conjugate transpose, based on \eqref{aa1}, and \eqref{bb2},  we derive the following result:
\begin{equation}
\left( \frac{1}{{M_x M_y }} \right) \sum_{n=0}^{M_x-1} e^{j\pi n(\mu_i^\varphi - \mu_k^\varphi)} \sum_{m=0}^{M_y-1} e^{j\pi m(\mu_i^h - \mu_k^h)} = 0. \label{eq:orthogonality_expanded}
\end{equation}
The geometric series summation formula is applied to each summation term separately, transforming the expression to the following conditions for zero interference:
\begin{equation}
e^{  j\pi M_x  \left( {\mu _i^\varphi   - \mu _k^\varphi  } \right)}  = 1,\; \; or\; \; e^{  j\pi M_y \left( {\mu _i^h  - \mu _k^h } \right)}  = 1,
\label{eq:zero_condition}
\end{equation}
indicating that $\left( {\mu _i^\varphi   - \mu _k^\varphi  } \right)$ and $ \left( {\mu _i^h  - \mu _k^h } \right)$ must be an integer multiple of 2.
Further analysis reveals that if \(M_x (\mu_i^\varphi - \mu_k^\varphi)\) and \(M_y (\mu_i^h - \mu_k^h)\) are both integer multiples of 2, then the orthogonality condition is satisfied. Rearranging these conditions, we get
\begin{equation}
\mu_i^\varphi - \mu_k^\varphi = \frac{2q_{i,k}^\varphi}{M_x}, \quad \mu_i^h - \mu_k^h = \frac{2q_{i,k}^h}{M_y}, \label{eq:mu_conditions}
\end{equation}
where \(q_{i,k}^\varphi\) and \(q_{i,k}^h\) are integers that ensure the difference in azimuth and elevation phases between any two users aligns with the array's geometry, thus guaranteeing orthogonality and effective spatial resource allocation. Now, we relax \eqref{eq:mu_conditions} as
\begin{equation}
\mu _i^\varphi   - \mu _k^\varphi   = \frac{{2q_{i,k}^\varphi  }}{{M_x  }} + \Delta  \; \; \; {\rm{or}} \; \; \;\mu _i^h   - \mu _k^h   = \frac{{2q_{i,k}^h  }}{{M_y }} + \Delta, \; \; \; i \ne k,
\label{g10}
\end{equation}
where the value of $\Delta $ is considered to be sufficiently small so that equation \eqref{g10} can still be valid. As a result, we can disregard the influence of co-channel interference among the users within the same resource block. In addition, by using \eqref{g10}, we can ascertain the minimum area required for each user to support orthogonal channels. In the azimuth and elevation domains, this minimum area is dictated by \(\Delta  \times \Delta \). As illustrated in Fig. \ref{system_model_figure}, the area that belongs to each sector is divided into degree of freedom (DoF) sections, and each section is further divided into $\Delta  \times \Delta $ subsections.
To determine the optimal values of \(\Delta \), we must solve an optimization problem. High values of \(\Delta \) result in increased interference among users but allow for efficient utilization of resource blocks. If one resource block is allocated to each user, the division of each subsection by the total number of resource blocks ensues, where \(\Delta\) corresponds to the width of these subsections divided by the number of resource blocks. With larger \(\Delta\) widths, users will be allocated more resource blocks, but interference among users will also increase. In the resource allocation problem, we determine the minimum number of resource blocks that should be allocated for each user in order to maximize the total rate while respecting the QoS conditions.
In order to derive the number of sections, we need to calculate the maximum DoF in both azimuth and elevation domains that are shown by \(N_\varphi\) and \(N_\theta\), respectively. 
In the following proposition, \(N_\varphi\) and \(N_\theta\) are derived based on the antenna configuration and the angular spread in each domain.
\begin{proposition}\label{proposition-outage}
The maximum DoF in azimuth domain (\(N_\varphi\)) is given by
\begin{equation}
 N_\varphi = \left\lfloor { M_x \sin \frac{\delta}{2}} \right\rfloor, 
 \label{azimuthDoF}
\end{equation}
 and  the maximum DoF in elevation domain \((N_\theta)\) is derived as
\begin{equation}
 N_\theta = \left\lfloor \frac{M_y \sin \theta}{2} \right\rfloor.
 \label{elevationDoF}
\end{equation}
\end{proposition}

\begin{proof}
See Appendix \ref{appendix2-outage}.
\end{proof}

\subsection{Achievable rate with the proposed resource allocation }
  
The transmitted signals of users in cluster \(l\) from the sector \(n\) of the HAPS are modulated by a precoding matrix \(\mathbf{P}_l^n \in \mathbb{C}^{M \times M}\) and then transmitted over the radio channel. Therefore, the received signal for the user in the \(l^{\text{th}}\) cluster of the \(m^{\text{th}}\) section at sector $n$ can be expressed as follows:
 \begin{equation}
\begin{aligned}
{\mathbf{ y}}_{m,l}^n  = \left( {{\bf h}_{m,l}^n } \right)^H  {\mathbf P_l^n \mathbf s_l^n} + {\mathbf n}_{m,l}^n,
\label{eq5}
\end{aligned}
\end{equation}
where \({\bf s}_l^n\) $\in \mathbb{C}^{M \times 1}$ represents the signal for cluster $l$ using the same resource block. 
$ {\mathbf n}_{m,l}^n $ is the additive white Gaussian noise with variance $ \sigma ^2$. 
 To further examine the user's observed signal, we can express it as follows:
\begin{equation}
\begin{array}{l}
 {\bf y}_{m,l}^n  = \left( {{\bf h}_{m,l}^n } \right)^H {\bf p}_{m,l}^n \sqrt {P_{\max } \Omega _{m,l}^n } s_{m,l}^n  +  \; \; \; \;\; \; \; \;\;   \\ 
 \underbrace {\sum\limits_{k = 1,k \ne l}^M {\left( {{\bf h}_{m,l}^n } \right)^H } {\bf p}_{m,k}^n \sqrt {P_{\max } \Omega _{m,k}^n } s_{m,k}^n }_{{\rm interference \;from \;  other\;   users\;   in\;  the \; cluster}} + {\bf n}_{m,l}^n . \\ 
 \end{array}
 \label{eq6}
\end{equation}
The $\Omega_{m,l}^n$ denotes the power allocation coefficient for the user in subsection $l$ of section $m$ served by sector $n$, and $s_{m,l}^n$ is the symbol for the $l^{\rm th}$ user in the $m^{\rm th}$ section.
$P_{\rm max}$ denotes the maximum power of the HAPS BS for serving users.
Inter-cluster interference is effectively eliminated by employing separate resource blocks across different clusters. Furthermore, inter-user interference is mitigated by ensuring that the channel gains of users within each cluster are orthogonal to each other. This orthogonalization enables the separation of user signals by multiplying the channel gain of each user with the received signal.
The data rate of user \( (m,l) \), to whom \( r \) resource blocks are allocated, is given by:
\begin{equation}
R_{m,l}^n   = r . BW_{RB} .  \log _2 \left( {1 + \frac{{\rho \Omega _{m,l}^n \left| {\left( {{\bf h}_{m,l}^n } \right)^H {\bf p}_{m,l}^n } \right|^2 }}{{\rho \sum\limits_{k,k \ne l} {\Omega _{m,k}^n \left| {{\bf h}_{m,l}^H {\bf p}_{m,k}^n } \right|^2  + 1} }}} \right),
 \end{equation}
 where $ \rho = \frac{P_{\max}}{N_0 . r . BW_{RB}}
 $.

\section{Problem Formulation}
The goal of this paper is to find the power allocation coefficients of terrestrial users to maximize the total rate while ensuring that each user achieves a pre-defined minimum data rate. The problem can be formulated as follows:
\begin{maxi!}|s|[2]
	{\Omega_{m,l}^n, r }{ {\sum\limits_{n = 1}^{N_s } {\sum\limits_{m = 1}^{N} {\sum\limits_{l = 1}^L {R_{m,l}^n } } } } \label{eq:ObjectiveExample3}}
	{\label{eq:Example3}}
    {}
\addConstraint{R_{m,l}^n  \ge R^{\rm QoS} ,\; \; \; r \in \left\{ {1, \cdots ,nbr } \right\}\label{eq:constr1}}
\addConstraint {P_{\max } {\sum\limits_{n = 1}^{N_s } {\sum\limits_{m = 1}^{N} {\sum\limits_{l = 1}^L { {\Omega _{m,l}^n } } }  \le P_t}} \label{eq:constr2}}
\addConstraint {
\Omega _{m,l}^n  \ge 0, \; \; \; \forall m,l,n.
\label{eq:constr3}}
\end{maxi!}
\begin{algorithm}
  \caption{ The resource allocation algorithm.}
\begin{algorithmic}[1] 
\State \textbf{ Initialization \& definitions.} \\
Initialize: $M_x, M_y, N_s, \delta $ and $r$.\\
Calculate $N_{\varphi}$, $N_{\theta} $ from  \eqref{azimuthDoF}, \eqref{elevationDoF}.
\\
 $N = N_\varphi   \times N_\theta$ (total sections per sector).
\\
 $nbr = \frac{{\rm BW}}{{\rm BW_{\rm RB} }}$ (number of resource blocks).\\
 $L = \frac{{nbr}}{r}$ (number of subsections per section).
\\
\textbf{ User clustering \& resource allocation.}
\For{$n=1$ to $N_s$}
\For{$l = 1$ to $L$} 
    \For{$m = 1$ to $N$} 
        \State $\rm{Cluster} \;l \gets {\rm{UE}}_{m,l}^n $  (section $m$, subsection $l$, sector $n$). 
    \EndFor
    \State Allocate $r$ resource blocks to $\rm{Cluster} \; l$.
    \If{there are multiple users at the same location $(m,l)$ in $\rm{Cluster} \; l$}
        \State Allocate orthogonal time slots to these users.
    \EndIf
\EndFor
\EndFor
\\ 
\textbf{ Power allocation to maximize the total rate.}\\
 Allocation of minimum power to satisfy the QoS constraint for all users.\\
$  \Omega _{m,l}^{n,\rm{min}}  = \left( {\frac{2^{R_{\rm{min}} }  - 1}{{\rho {\left| { \left( {{\bf h}_{m,l}^n } \right)^H  {\mathbf p}_{m,l} } \right|^2} }}} \right),  $\\
Arrange users according to their channel quality gains based on the following equation:
\\
$ \Delta P_{m,l}^n  = \left( {2^{\Delta R}  - 1} \right)\frac{{P_{\max } 2^{R_{m,l}^n } }}{{\rho \left| {\left( {{\bf h}_{m,l}^n } \right)^H {\bf p}_{m,l} } \right|^2 }}$.\\
Sort fractional levels in ascending order:\\
 $ \left[ {H,m} \right] = {\rm{sort}}\left( {\frac{{P_{\max } 2^{ {R_{\min } } } }}{{\rho \left| {\left( {{\bf h}_{m,l}^n } \right)^H {\mathbf p}_{m,l} } \right|^2}}} \right)$.\\
    ${\bf while} \; ({\rm P}_{{\rm rem}}  > 0):$ \\ 
    $\left\{ {{\Omega _{m\left( i \right),l}^{n } }  = \frac{{H\left( {i + 1} \right) - H\left( i \right)}}{{P_{\max } }},\; \; i +  = 1} \right\}$.
\\  
\textbf{ Calculate:}\\
$
R_{m,l}^n   = r . {\rm BW}_{\rm RB}  . \log _2 \left( {1 + \frac{{\rho \Omega _{m,l}^n \left| {\left( {{\bf h}_{m,l}^n } \right)^H {\bf p}_{m,l}^n } \right|^2 }}{{\rho \sum\limits_{k,k \ne l} {\Omega _{m,k}^n \left| {{\bf h}_{m,l}^H {\bf p}_{m,k}^n } \right|^2  + 1} }}} \right)
$.
\label{al1}
\end{algorithmic}
\end{algorithm}
Constraint \eqref{eq:constr1} ensures the QoS conditions for users, while \eqref{eq:constr2} and \eqref{eq:constr3} represent the power limitations of the HAPS and users, respectively.
The objective function of equation \eqref{eq:ObjectiveExample3} is non-convex, and hence obtaining an optimal solution for this problem is challenging. 
In Algorithm 1, a resource allocation algorithm is presented for distributing power, frequency, and time. Initially, users are grouped according to their channel directions, and resource blocks are distributed among these clusters. In cases where multiple users are situated within the same subsection, time slot allocations are applied to these users. 
As the width and length of the subsection increase, the orthogonality among users sharing the same resource blocks decreases, resulting in increased interference among them. Therefore, selecting the optimal number of resource blocks involves a trade-off between spectral efficiency and interference management. 
The number of resource allocation states is constrained by the total number of resource blocks available. Therefore, in the simulation results, we select the most favorable allocation state according to maximizing the total rate.
Subsequently, for power allocation, we initially assign the minimum power necessary to meet the QoS requirements of users. Then, the remaining power is allocated among users to maximize the total rate. To allocate the remaining power, users are sorted based on the quality of their channel gain, with higher priority given to users with superior channel gain. This ensures that more power is allocated to users with better channel conditions, optimizing overall system performance.

\section{Simulation results}
This section provides simulation outcomes for the proposed scenario, illustrating its performance improvements. The simulation parameters are detailed in table \ref{table:simulation-parameters}.
We explore the spectral efficiency of a massive MIMO HAPS BS utilizing our proposed antenna model. In this scenario, we tackle the integrated resource and power allocation problem to maximize the total data rate. Since the spatial correlation among adjacent users is significantly high in massive MIMO HAPS systems, different resource blocks are allocated to users who interfere with each other to mitigate interference.
 Users with orthogonal channels in azimuth, elevation, or both receive the same resource blocks.  The boundary of orthogonality thus plays a crucial role in defining user interference. 
 \begin{table}[h]
\centering
\caption{Simulation parameters.}
\begin{tabular}{cc} 
\hline\hline 
Area radius & $ 100\; \rm{km} $ \\
Carrier frequency $(f_c)$ & $ 2.5 \; \rm{GHz} $ \\
Bandwidth & $ 10\; \rm{MHz}, 20\; \rm{MHz}$ \\
Bandwidth per resource block & $180 \; \rm{kHz}$ \\
Elevation of the HAPS & $ 20\; \rm{km}$ \\
Minimum rate of each user ($R_{\min}$) &$ 1\; {{\rm{bps}} \mathord{\left/
 {\vphantom {{\rm{bps}} {\rm{Hz}}}} \right.
 \kern-\nulldelimiterspace} {\rm{Hz}}} $ \\
Power spectral density of thermal noise & $  - 174 \; \rm{dBm/Hz}  $ \\
Noise figure & $ 7\; \rm{dB}$\\
Shadow fading parameters ($\sigma _{\text{SF}}^{\text{LoS}}, \sigma _{\text{SF}}^{\text{NLoS}}$) & $(4,6) \, \rm{dB} $ \cite{3gpp.36.33} \\
\hline 
\end{tabular}
\label{table:simulation-parameters}
\end{table}

 Fig.  \ref{datarate10MHZ} illustrates the sum rate of a massive MIMO HAPS system with a constant total number of cell sections and varying numbers of subsections. We conducted this analysis with two different bandwidths, i.e., \(10 \text{ MHz}\) and \(20 \text{ MHz}\). For a bandwidth of \(10 \text{ MHz}\), the highest sum rate occurs with \(L=25\) subsections and $r=2$ resource blocks allocated per user. With more subsections, such as \(L=49\), spectral efficiency decreases due to under-utilization bandwidth, while fewer subsections, such as \(L=16\), result in higher interference among users.
This condition applies to a bandwidth of \(20 \text{ MHz}\), where the highest total rate is achieved with \(L=36\) subsections and \(r=3\) resource blocks per user.
\begin{figure}[htbp]
\centerline{\includegraphics[width =8.5cm]{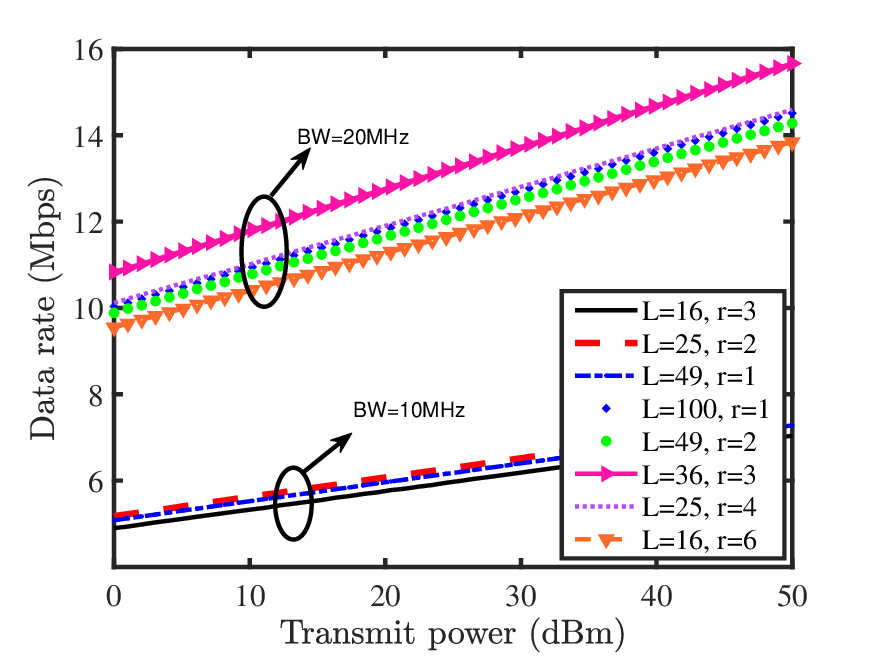}}
\caption{Total sum rate versus HAPS BS power with \(M=16\) elements per UPA, and antenna spacing \(0.5\lambda\). Variables \(L\) and \(r\) indicate the number of subsections per section and resource blocks per user, respectively. The total number of resource blocks is \(50\) for a bandwidth of \(10 \text{ MHz}\) and \(100\) for a bandwidth of \(20 \text{ MHz}\). }
\label{datarate10MHZ}
\end{figure}
\begin{figure}[htbp]
\begin{centering}
\centerline{\includegraphics[width=8.5cm]{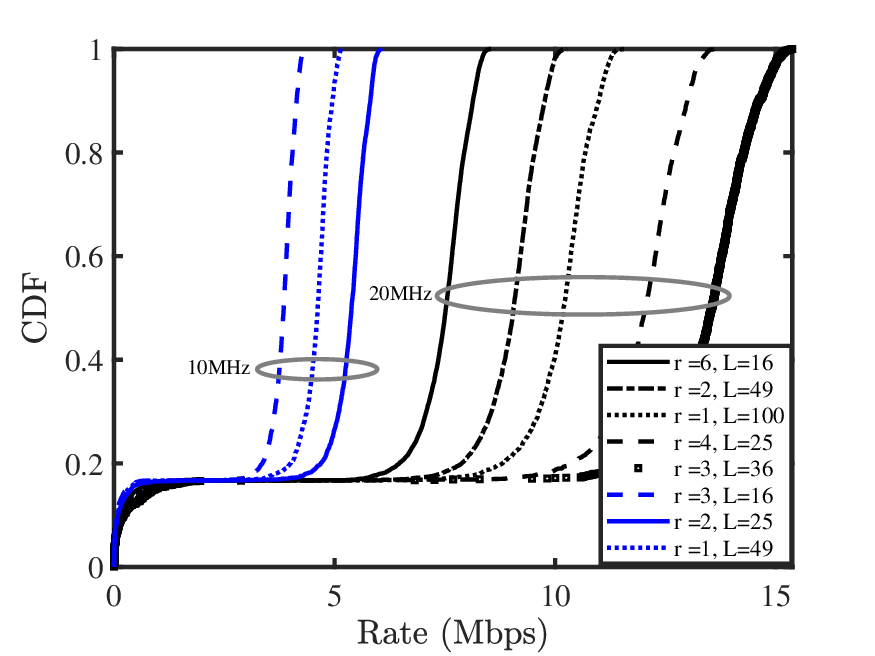}}
\caption{Comparison of resource blocks and subsections for \(10~\text{MHz}\) and \(20~\text{MHz}\) bandwidths. For \(10~\text{MHz}\) (blue lines), resource blocks of 1, 2, and 3 correspond to subsections of 49, 25, and 16, respectively. For \(20~\text{MHz}\) (black lines), resource blocks of 1, 2, 3, 4, and 6 correspond to subsections of 100, 49, 36, 25, and 16, respectively. Parameters used are \(N_s = 6\), \(M_x = 3\), and \(M_y = 4\).
}
\label{CDF2}
\end{centering}
\end{figure}

Fig.~\ref{CDF2} shows the cumulative distribution function (CDF) of the rate across different numbers of resource blocks and subsections for two bandwidth scenarios: \(10~\text{MHz}\) and \(20~\text{MHz}\). For the \(10~\text{MHz}\) bandwidth (blue lines), the resource blocks are set to 1, 2, and 3, corresponding to subsection counts of 49, 25, and 16, respectively. For the \(20~\text{MHz}\) bandwidth (black lines), the resource block values are 1, 2, 3, 4, and 6, with subsection counts of 100, 49, 36, 25, and 16, respectively.
As the number of resource blocks increases, the rate performance improves, with the \(20~\text{MHz}\) bandwidth achieving higher rates than the \(10~\text{MHz}\) bandwidth, as evidenced by the rightward shift in the CDF curves. Notably, the configurations \(L=25\) with \(r=2\) for \(10~\text{MHz}\) and \(L=36\) with \(r=3\) for \(20~\text{MHz}\) deliver optimal performance, illustrating the balance between spectral efficiency and interference management. This result emphasizes the trade-off between efficiently using bandwidth and managing interference across users.


 Fig. \ref{heatmap}  visualizes the correlation coefficients among users sharing resource blocks within the same cluster.  The varying color intensities on this heatmap illustrate the degrees of correlation among users. The amount of numbers indicates that users with orthogonal channel gains are selected appropriately. It is evident that as the number of subsections increases, the correlation among users decreases.
\begin{figure}[htbp]
\centerline{\includegraphics[width =6.3cm]{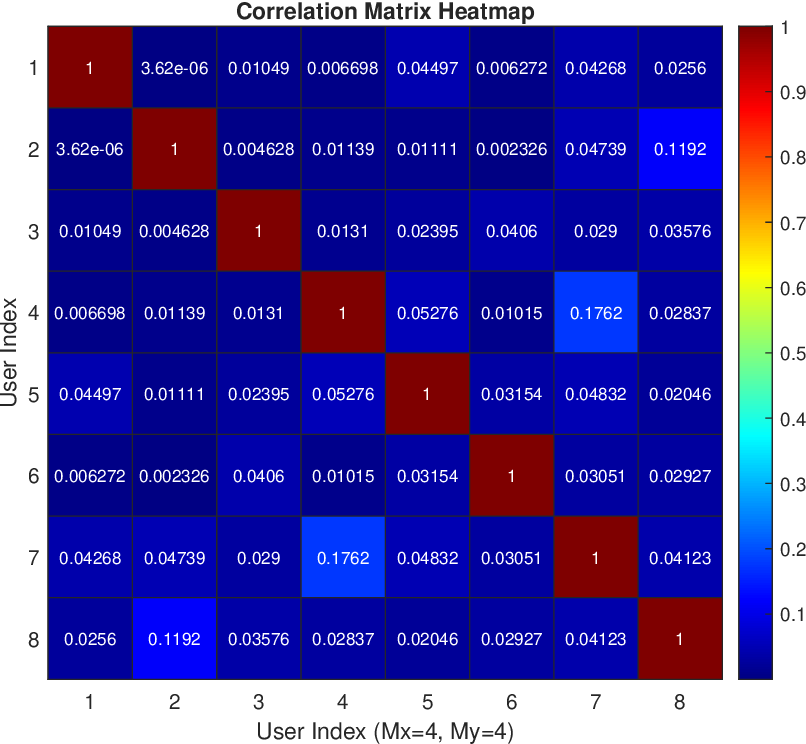}}
\caption{Heatmap of the correlation coefficients among users sharing resource blocks within the same cluster. }
\label{heatmap}
\end{figure}

\section{Conclusions}
This study introduced a novel sectorized antenna architecture and a strategic spatial resource allocation technique to enhance massive MIMO systems for the HAPS. By utilizing UPAs for reduced spatial correlation and dynamic resource block allocation based on channel orthogonality, our approach effectively minimizes interference and maximizes data rates. 
Simulation findings highlighted how spatial resource allocation influences the management of spatial correlation and user interference.

\appendices
\section{Proof of Proposition \ref{proposition-outage}}
\label{appendix2-outage}
As indicated in Fig. \ref{fig:main_figure}, for a limited azimuth angle of UPA $(\delta)$, the range of $\mu _i^{\varphi} - \mu _k^{\varphi}$ in \eqref{g10} changes to $2\sin(\delta/2)$, and the range of $\mu_i^{\theta} - \mu_k^{\theta}$ is $\sin \left( \theta \right)$, where $\theta = \arctan \left( \frac{r}{h} \right)$ and $\delta = \frac{2\pi}{N_s}$.
Thus, the maximum degrees of freedom in each azimuth and elevation axis for a vertical UPA with $\delta$ azimuth angular are derived from these equations:
\begin{equation}
\begin{array}{l}
N_\varphi = \left[ \frac{2 \sin \frac{\delta}{2}}{\frac{2}{M_x}} \right] = \left[  M_x \sin \frac{\delta}{2} \right], \\
N_{\theta} = \left[ \frac{\sin \theta}{\frac{2}{M_y}} \right] = \left[ \frac{M_y \sin \theta}{2}\right]. \\
\end{array}
\end{equation}
Therefore, the maximum degree of freedom for a UPA sector with an azimuth angle $\delta$ is given by
\begin{equation}
N = N_\varphi  N_{\theta}.
\end{equation}

\begin{figure}[htbp]  
    \centering
    \begin{subfigure}[b]{0.18\textwidth}
        \centering
        \includegraphics[width=0.9\textwidth]{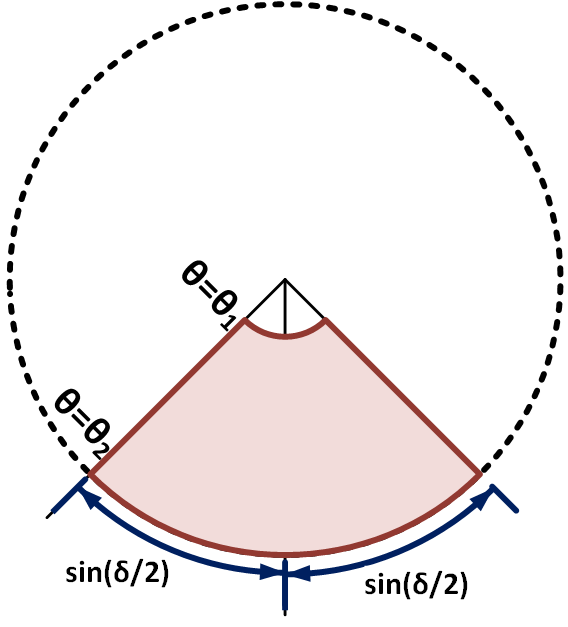}  
        \caption{
        }
        \label{a}
    \end{subfigure}
    \hspace{0\textwidth}  
    \begin{subfigure}[b]{0.25\textwidth}
        \centering
        \includegraphics[width=0.9\textwidth]{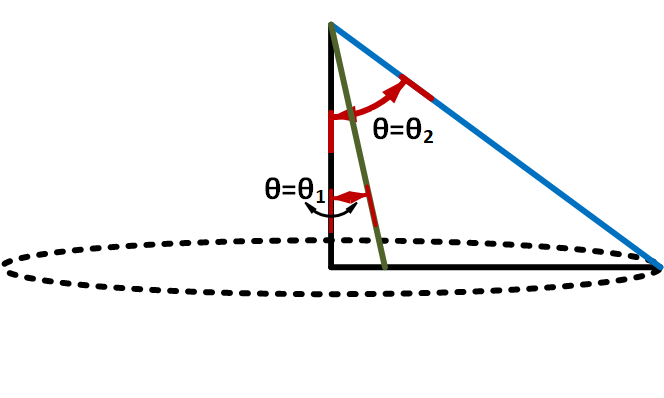}
         \caption{
         }
        \label{b}
    \end{subfigure}
    \caption{Visualization of the (a) azimuth and (b) elevation angles associated with each UPA sector.}
    \label{fig:main_figure}
\end{figure}
\ifCLASSOPTIONcaptionsoff
  \newpage
\fi
\bibliographystyle{IEEEtran}
\bibliography{IEEEabrv,Bibliography}
\vfill
\end{document}